\documentclass{cccg24}
\usepackage{graphicx,amssymb,amsmath}
\usepackage{todonotes}

\newcommand{\dist}{\operatorname{dist}}

\title{Hyperplane Distance Depth\thanks{This work is funded in part by the Natural Sciences and Engineering Research Council of Canada (NSERC).}

\author{
        Stephane Durocher\thanks{Department of Computer Science, University of Manitoba, {\tt \{stephane.durocher,amirhossein.mashghdoust\}@umanitoba.ca}}
        \and
        Amirhossein Mashghdoust\footnotemark[2]}
}

\begin{document}
\thispagestyle{empty}
\maketitle

\begin{abstract}
Depth measures quantify central tendency in the analysis of statistical and geometric data. Selecting a depth measure that is simple and efficiently computable is often important, e.g., when calculating depth for multiple query points or when applied to large sets of data. In this work, we introduce \emph{Hyperplane Distance Depth (HDD)}, which measures the centrality of a query point $q$ relative to a given set $P$ of $n$ points in $\mathbb{R}^d$, defined as the sum of the distances from $q$ to all $\binom{n}{d}$ hyperplanes determined by points in $P$. We present algorithms for calculating the HDD of an arbitrary query point $q$ relative to $P$ in $O(d \log n)$ time after preprocessing $P$, and for finding a median point of $P$ in $O(d n^{d^2} \log n)$ time. 
We study various properties of hyperplane distance depth, and show that it is convex, symmetric, and vanishing at infinity.
\end{abstract}

\section{Introduction}
Depth measures describe central tendency in statistical and geometric data. A median of a set of univariate data is a point that partitions the set into two halves of equal cardinality, with smaller values in one part, and larger values in the other. Various definitions of medians exist in higher dimensions (multivariate data), seeking to generalize the one-dimensional notion of median (e.g., \cite{durocher2017}). For geometric data and sets of geometric objects, applications of median-finding include calculating a centroid, determining a balance point in physical objects, and defining cluster centers in facility location problems \cite{farahani2010multiple}.
A median is frequently used in statistics to describe the 
central tendency of a data set. It is particularly useful when dealing with skewed distributions or datasets that contain outliers. By using a  median, analysts can obtain a representative value that is less affected by extreme values and outliers \cite{murray1998cluster}.

In 1975, Tukey introduced the concept of data depth for evaluating centrality in bivariate data sets \cite{tukey1975}. The depth of a particular query point $q$ in relation to a given  set $P$ gauges the extent to which $q$ is situated within the overall distribution of $P$; i.e., when $q$'s depth is large, $q$ tends to be near the center of $P$.
Since the introduction of Tukey depth (also called half-space depth), many more depth functions have been proposed. 

Data depth functions should ideally satisfy specific properties, such as \emph{convexity}, \emph{stability} (small perturbations in the data do not result in large changes in depth values),
\emph{robustness} (depth is not heavily influenced by outliers or extreme values in the data), \emph{affine invariance} (the depth function remains consistent under linear transformations of the data, such as translation, scaling, and rotation), \emph{maximality at the center} (points closer to the geometric center of the data set have higher depth values), and \emph{vanishing at infinity} (depth values approach zero as a query point moves away from the data set) \cite{zuo2000general}.

\section{Related Work}\label{sec:relatedWork}
Tukey \cite{tukey1975} first introduced the concept of location depth. In $\mathbb{R}^2$, the Tukey depth of a point $q\in{R}^2$ relative to a set $P$ of $n$ points in $\mathbb{R}^2$ is defined as the smallest number of points of $P$ on one side of a line passing through $q$. This concept can also be generalized to higher dimensions.

\newtheorem{definition}{Definition}
\begin{definition}[Tukey Depth \cite{tukey1975}]
The \emph{Tukey depth} of a point $q \in \mathbb{R}^d$ relative to a set $P$ of points in $\mathbb{R}^d$, is the minimum number of points of $P$ in any closed half-space that contains $q$.
\end{definition}
In univariate space, e.g., in $\mathbb{R}$, the Tukey depth of $q$ is determined by considering the minimum of the count of points $p_i \in P$ where $p_i<q$, and the count of points $p_i \in P$ where $p_i>q$.\\
A \emph{Tukey median} of a set $P$ in $\mathbb{R}^d$ corresponds to a point (or points) with maximum Tukey depth among all points in $\mathbb{R}^d$.

Since Tukey's introduction of Tukey depth, several other important depth functions have been defined to measure the centrality of $q$ relative to $P$.

\begin{definition}[Mahalanobis Depth \cite{mahalanobis2018generalized}]
The \emph{Mahalanobis dept}h of a point $q \in \mathbb{R}^d$ relative to a set ${P}$ in $\mathbb{R}^d$ is defined as ${[1+{(q-\bar{q})}^T{P_d}^{-1}(q-\bar{q})]}^{-1}$, where $\bar{q}$ and $P_d$ are the mean vector and dispersion matrix of $P$.
\end{definition}
This function lacks robustness, as it relies on non-robust measures like the mean and the dispersion matrix. Another possible disadvantage of Mahalanobis depth is its reliance on the existence of second moments \cite{mahalanobis2018generalized}.

\begin{definition}[Convex Hull Peeling Depth \cite{barnett1976ordering}]
The \emph{convex hull peeling depth} of a point $q \in \mathbb{R}^d$ relative to a set ${P}$ in $\mathbb{R}^d$ is the level of the convex layer to which $q$ belongs.
\end{definition}
A convex layer is established by recursively removing points on the convex hull boundary of $P$ until $q$ is outside the hull. Begin by constructing the convex hull of $P$. Points of $P$ on the boundary of the hull constitute the initial convex layer and are removed. Then, form the convex hull anew with the remaining points of $P$. The points along this new hull's boundary constitute the second convex layer. This iterative process continues, generating a sequence of nested convex layers. The deeper a query point $q$ resides within $P$, the deeper the layer it belongs to.
However, the method of convex hull peeling depth possesses certain drawbacks. It fails to exhibit robustness in the presence of outliers or noise. Additionally, it's unfeasible to associate this measure with a theoretical distribution.

\begin{definition}[Oja Depth \cite{oja1983descriptive}]
The \emph{Oja depth} of a point $q \in \mathbb{R}^d$ relative to a set ${P}$ in $\mathbb{R}^d$ is defined as the sum of the volumes of every closed simplex having one vertex at $q$ and its remaining vertices at any points of $P$.
\end{definition}
In $\mathbb{R}^2$, the Oja depth of a point $q$ is the sum of the areas of all triangles formed by the vertices $q$,$p_i$, and $p_j$, where  $\{p_i, p_j\} \subseteq P$.

\begin{definition}[Simplicial Depth \cite{liu1990notion}]
The \emph{simplicial Depth} of a point $q \in \mathbb{R}^d$ relative to a set ${P}$ in $\mathbb{R}^d$ is defined as the number of closed simplices containing $q$ and having $d+1$ vertices in $P$.
\end{definition}
The simplicial depth of a point $q \in \mathbb{R}^2$ is the number of triangles with vertices in $P$ and containing $q$. This is a common measure of data depth.

\begin{definition}[$L_1$ Depth \cite{vardi2000multivariate}]
The \emph{$L_1$ depth} of a point $q \in \mathbb{R}^d$ relative to a set ${P}$ in $\mathbb{R}^d$ is defined as $\sum_{p_i\in P} {||p_i-q||}_1$.
\end{definition}
The $L_1$ Median is the point that minimizes the sum of the absolute distances (also known as the $L_1$ norm or Manhattan distance) to all other points in $P$.
The key advantage of the $L_1$ Median is its robustness to outliers. It is less sensitive to extreme values in the dataset compared to the $L_2$ Median, which minimizes the sum of squared distances. As a result, the $L_1$ Median can provide a more accurate estimate of central tendency in datasets with outliers or heavy-tailed distributions.
The $L_1$ Median is used in various fields, including finance, image processing, and robust statistics, whenever there is a need for a robust estimate of the central location of a dataset that may contain atypical values.

\begin{definition}[$L_2$ Depth \cite{zuo2000general}]
The \emph{$L_2$ depth} (mean) of a point $q \in \mathbb{R}^d$ relative to a set ${P}$ in $\mathbb{R}^d$ is defined as $\sum_{p_i\in P} {||p_i-q||}^2$.
\end{definition}
The $L_2$ Median is the point that minimizes the sum of the squared Euclidean distances.
The mean is a widely used measure of central tendency in statistics and data analysis. The mean is not robust to outliers; a single outlier can pull the mean arbitrarily far.

\begin{definition}[Fermat-Weber Depth \cite{durier1985geometrical}]
The \emph{Fermat-Weber depth} (Geometric depth) of a point $q \in \mathbb{R}^d$ relative to a set ${P}$ in $\mathbb{R}^d$ is defined as $\sum_{p_i\in P} {||p_i-q||}$.
\end{definition}
A deepest point (median) with respect to Fermat-Weber depth cannot be calculated exactly in general when $d \geq 2$ and $|P| \geq 5$ \cite{bajaj1988}.

There is no single depth function that universally outperforms all others. The choice of a particular depth function often depends on its suitability for a specific dataset or its ease of computation. Nevertheless, there are several desirable properties that all data depth functions should ideally possess. In Section~\ref{sec:results}, we introduce a new depth measure, and we examine which of these properties it satisfies.

\section{Results}\label{sec:results}
In this section, we will introduce the Hyperplane Distance Depth (HDD) measure and study its properties.

\subsection{Defintion}
 
\begin{definition}[Hyperplane distance depth]\label{Hyperplane distance depth}
The Hyperplane distance depth (HDD) of a point $q \in \mathbb{R}^d$ relative to a set ${P}$ in $\mathbb{R}^d$ is defined as
\begin{equation}\label{hddDefinition}
    D_P(q)=\sum_{h_i \in H_P} \dist(q,h_i),
\end{equation} 
where $H_P$ is the set of all $\binom{n}{d}$ $(d-1)$-dimensional hyperplanes determined by points in $P$, and $\dist(q,h_i)$ denotes the Euclidean ($L_2$) distance from the point $q$ to the hyperplane $h_i$.
\end{definition}

Both Fermat-Weber depth and hyperplane distance depth are defined as sums of Euclidean ($L_2$) distances. 
Unlike Fermat-Weber depth, for which the location of a median cannot be computed exactly in general when $d \geq 2$ \cite{bajaj1988}, as we show in Section~\ref{Algorithms}, the location of a HDD median can be computed exactly. 

\subsection{Properties}\label{sec:properties}
\begin{theorem} \label{1D}
In $\mathbb{R}$, the HDD median relative to the set ${P}$ coincides with the usual univariate definition of median.
\end{theorem}
\begin{proof}
By Definition~\ref{Hyperplane distance depth}, the median is a point that minimizes the sum of the distances to all possible points passing through each point in $P$. Therefore, $H_P=P$. Consequently, the HDD median is equivalent to the usual definition of median in a one-dimensional space.
\end{proof}

\begin{theorem} \label{convexity}
The HDD function $D_P(q)$ relative to the set ${P}$ is convex over $q \in \mathbb{R}^d$.
\end{theorem}
\begin{proof}
The distance function ${d_{h_i}(q)}$ from a query point $q$ to the hyperplane $h_i$ is convex. Any non-negative linear combination of convex functions is convex. Therefore, the HDD function  $\sum_{h_i \in H_P} {d_{h_i}(q)}=D_P(q)$ is convex over $q$.
\end{proof}
\begin{theorem} \label{medianOnIntersections}
The HDD median point relative to the set ${P}$ of points in $\mathbb{R}^d$ is always on one of the intersection points between $d$ hyperplanes in $H_P$.
\end{theorem}
\begin{proof}
The distance from the point $q \in \mathbb{R}^d$ to a hyperplane $h_i$ is equal to ${d_{h_i}(q)}=\frac{|{w_i}.q+{b_i}|}{\|{w_i}\|}$ where $w_i$ and $b_i$ are the hyperplane's normal vector and the offset respectively. Therefore, the HDD of the point $q$ is equal to
\begin{equation} \label{eq1}
D_P(q)=\sum_{h_i \in H_P} {d_{h_i}(q)}=\sum_{h_i \in H_P} \frac{|{w_i}.q+{b_i}|}{\|{w_i}\|}
\end{equation}
Depending on the position of $q$ with respect to $h_i$, ${d_{h_i}(q)}=\frac{|{w_i}.q+{b_i}|}{\|{w_i}\|}$ can be equal to $+\frac{{w_i}.q+{b_i}}{\|{w_i}\|}$ (above the hyperplane) ,$-\frac{{w_i}.q+{b_i}}{\|{w_i}\|}$ (below the hyperplane), or $0$ (on the hyperplane). Therefore, for any point $q$ we have
\begin{equation} \label{eq2}
\begin{split}
    & \sum_{h_i \in H_P} {d_{h_i}(q)}=\sum_{h_i \in H_P} g_{i,q} \frac{{w_i}.q+{b_i}}{\|{w_i}\|}\\
    &g_{i,q}= 
\begin{cases}
   +1,& \text{if $q$ is above $h_i$} \\
   -1,& \text{if $q$ is below $h_i$}\\
    0,& \text{if $q$ is on $h_i$}\\
\end{cases}
\end{split}
\end{equation}
It is worth noting that the derivative of the equation \eqref{eq2} exists if $q$ is not on any of the hyperplanes in $H_P$ ($g_{i,q}\neq0$). Now to find the HDD median with the minimum HDD measure, we should compute the derivative with respect to $q$ and see where it will be 
equal to $0$. For any query point $q$ inside a region bounded by some $H_P$ hyperplanes and not on any $H_P$ hyperplanes (Figure~\ref{fig:4PointsExample}) we have 
\begin{equation} \label{eq3}
\begin{split}
    \frac{d}{dq} D_P(q)& = \frac{d}{dq}\sum_{h_i \in H_P} g_{i,q} \frac{{w_i}.q+{b_i}}{\|{w_i}\|}\\& = \sum_{h_i \in H_P} g_{i,q} \frac{{w_i}}{\|{w_i}\|}
\end{split}
\end{equation}
\begin{figure}[h] 
\centering
\includegraphics[width=0.5\textwidth]{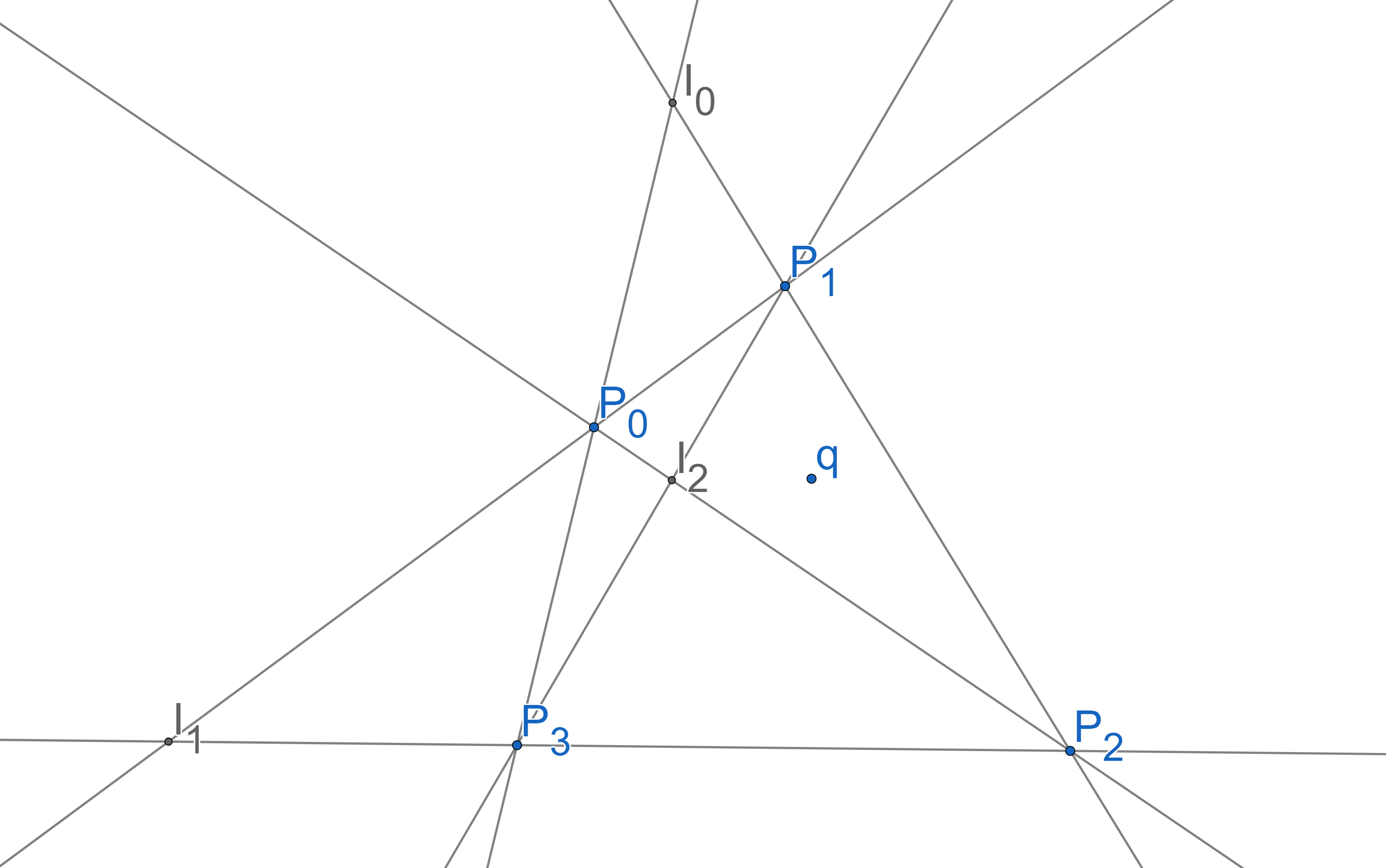}
\caption{Example of HDD in two dimensions: $P=\{P_0,P_1,P_2,P_3\}$ is the set of input points, $I=\{I_0,I_1,I_2\}$ is the set of intersection points, and $q$ is the query point.}
\label{fig:4PointsExample}
\end{figure}
Equation~\eqref{eq3} above cannot be equal to $0$ in general since there are no variables \eqref{eq3}. This means the assumption we made about the query point not being on the hyperplanes in $H_P$ was incorrect. Therefore, we can say the median is surely on one of the hyperplanes. If $q$ is on $h_j$, ${w_j}.q+{b_j}$ will be equal to $0$. Therefore, we can say
\begin{align}
    \frac{d}{dq} D_P(q)& =\frac{d}{dq}\sum_{h_i \in H_P-\{h_j\}} g_{i,q} \frac{{w_i}.q+{b_i}}{\|{w_i}\|} \nonumber \\
    & =\sum_{h_i \in H_P-\{h_j\}} g_{i,q} \frac{{w_i}}{\|{w_i}\|} . \label{eq5}
\end{align}
Using the same logic we can conclude that the median point should be on another hyperplane in addition to $h_j$. We can repeat these steps $d$ times and after that, it will be proved that the median should be on the intersection point of $d$ hyperplane (that will be a single point), thus the median will be on one of the intersection points.
\end{proof}
\begin{theorem} \label{MedianInCH}
The HDD median point relative to the set ${P}$ in $\mathbb{R}^d$ is always in the convex hull of the input points $P$.
\end{theorem}
\begin{proof}
Let $D'_{p_i}(q)$ be the sum of the distances to all the hyperplanes in $H_P$ passing through the point $p_i$. The minimum of this convex function is always on the point $p_i$ where the HDD is equal to $0$. On the other hand, since each hyperplane includes $d$ input points from $P$, we have $D_P(q)=d\sum_{p_i \in P}{D'_{p_i}(q)}$.\\
Now consider a point $q_o$ outside of the convex hull. by computing the gradient of $D_P(q_o)$, we will show that by moving $q_o$ closer to the convex hull, the HDD gets smaller. Using the equation above we have $-\nabla H_{P}(q_o)=-d\sum_{p_i \in P}\nabla{D'_{p_i}(q_o)}$. Since the minimum of the function $D'_{p_i}(q)$ is on $p_i$, $-\nabla D'_{p_i}(q)$ is a vector pointing to $p_i$ for $p_i \in P$. Therefore we can conclude that for every point $q_o$ outside of the convex hull, $-\nabla H_{P}(q_o)$ points to the convex hull that means by moving toward that direction, the HDD decreases. Therefore the HDD median is always in the convex hull of $P$.
\end{proof}
\begin{theorem} \label{simmetry}
The HDD median point relative to the set ${P}$ in $\mathbb{R}^d$ is always at the center of symmetry.
\end{theorem}
\begin{proof}
Let $p_M$ be the median of the $P$ s.t. $P$ is symmetric. If $p_M$ is not on the center of symmetry, consider ${p'}_M$, the reflection of $p_M$ across the center of symmetry. Because of the symmetry, it is trivial that the HDD measure of both points is equal. Since the median point has the minimum depth measure among the other points and the depth measure function is convex, all the points on the line segment $p_m{p'}_M$ should have a depth measure equal to the median. Therefore, the median is always at the center of symmetry.
\end{proof}
\begin{theorem} \label{vanishing}
The HDD measure relative to the set ${P}$ in $\mathbb{R}^d$ vanishes as we move the query point to infinity.
\end{theorem}
\begin{proof}
As we move the query point $q$ to infinity, it is straightforward that there exists a hyperplane $h_i \in H_P$ that gets further from $q$. Since we can move $q$ arbitrarily far from $h_i$, and the distance from $q$ to the remaining hyperplanes in $H_P$ is non-negative, therefore HDD vanishes at infinity.
\end{proof}

Note that some measures of depth are defined such that deep points have high depth values and outliers have low depth values, whereas this property is reversed for other depth measures. HDD is of the latter type, with central points having a low sum of distances to hyperplanes in $H_P$, whereas this sum approaches infinity as $q$ moves away from $P$. Consequently, for HDD, ``vanishing at infinity'' means that depth approaches $\infty$ as opposed to 0.

\begin{theorem} \label{robustness}
The HDD measure relative to the set ${P}$ in $\mathbb{R}^d$ is not robust.
\end{theorem}
\begin{proof}
We will prove this fact using a counter-example in a 2-dimensional space (Figure~\ref{fig:Property_Robustness}). We can move the HDD median by changing the location of 2 points which means the HDD is not robust. The median is always on one of the intersection points and we can place the points in a way that $I_0$ is always the median (Figure~\ref{fig:Property_Robustness}). We will compute the depth measures for the points $I_0$ \eqref{eq5} and $I_i$ \eqref{eq6}, where $I_i$ is an arbitrary intersection point except $I_0$.
\begin{figure}[htp]
    \centering
    \includegraphics[width=8cm]{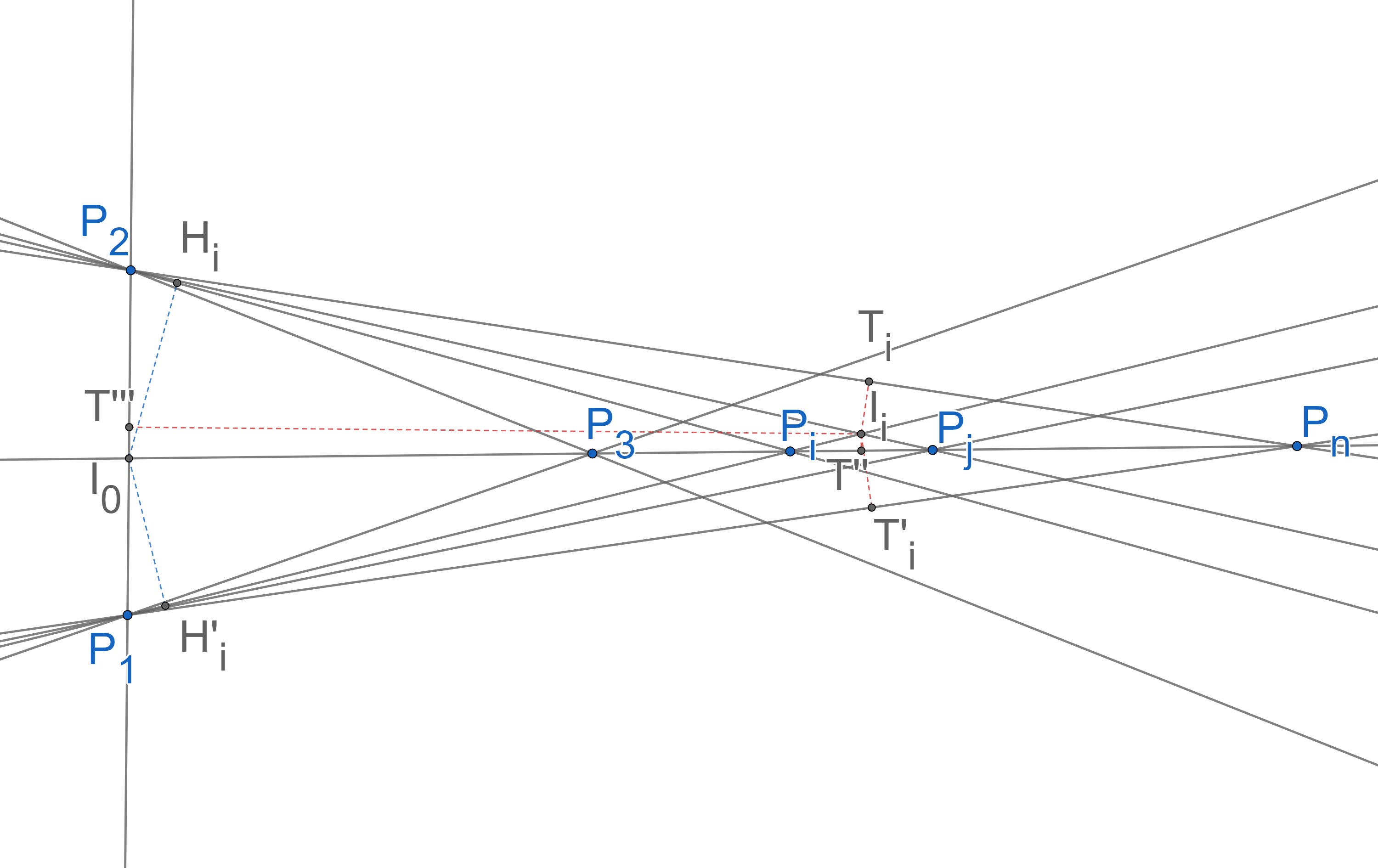}
    \caption{A counter-example that shows the HDD is not robust}
    \label{fig:Property_Robustness}
\end{figure}

\begin{align}
D_P(I_0) & = \sum_{i\in [3,n]}{ d_{l_{P_1P_i}}(I_0)}+\sum_{i\in [3,n]}{d_{l_{P_2P_i}}(I_0)} \label{eq10} \\
D_P(I_i)& =\sum_{i\in [3,n] }{d_{l_{P_1P_i}}(I_i)} + \sum_{i\in [3,n] }{d_{l_{P_2P_i}}(I_i)} \nonumber \\
& + \binom{n-2}{2} d_{l_{P_3P_n}}(I_i) + d_{l_{P_1P_2}}(I_i) \label{eq6}
\end{align}
Regardless of the $I_0$ position, we know that $I_0H_i<I_0P_2$ and $I_0H'_i<I_0P_1$. Therefore, we have (Equation~\eqref{eq10}):
\begin{equation} \label{eq7}
D_P(I_0)<(n-2)I_0P_2 + (n-2)I_0P_1=(n-2)P_1P_2 
\end{equation}
On the other hand, using Equation~\eqref{eq6} we have:
\begin{equation} \label{eq8}
D_P(I_i)>d_{l_{P_1P_2}}(I_i)
\end{equation}
Now by moving the points $P_1$ and $P_2$ far enough, let $d_{l_{P_1P_2}}(I_i)=(n-2)P_1P_2+m$, where $m$ is a positive number. Therefore, we have (inequality \ref{eq8}):
\begin{equation} \label{eq9}
D_P(I_i)>(n-2)P_1P_2+m
\end{equation}
Combining the inequality \ref{eq7} and \ref{eq9} we have $D_P(I_i)>D_P(I_0)$.

Consequently, $I_0$ is the median. By increasing $m$, the median $I_0$ gets as far as we want. This means by moving $P_1$ and $P_2$, we can move the median point as much as we want.
\end{proof}

\begin{definition}[$k$-stability \cite{durocher2009}]
A depth measure $D$ is \emph{$k$-stable} if for all points $q$ in $\mathbb{R}^d$, all sets $P$ in $\mathbb{R}^d$, all $\epsilon > 0$, and all functions $f_\epsilon:\mathbb{R}^d \to \mathbb{R}^d$ such that $\forall p$, $\dist(p,f(p)) \leq \epsilon$, 
\begin{equation}
    k \cdot |D(q,P) - D(f_\epsilon(q),f_\epsilon(P)| \leq \epsilon ,
\end{equation}
where $f_\epsilon(P) = \{ f_\epsilon(p) \mid p \in P\}$.
\end{definition}

That is, for any $\epsilon$-perturbation of $P$ and $q$, the depth of $q$ relative to $P$ changes by at most $k \epsilon$.

\begin{theorem}\label{thm:stability}
The HDD measure relative to the set ${P}$ in $\mathbb{R}^d$ is not $k$-stable for any constant $k$.
\end{theorem}

\begin{proof}
Choose any $k > 0$ and let $n = \max\{1, \lceil 1/k\rceil + 1\}$. Let $P$ be a set of $n$ points in $\mathbb{R}$ and let $q \in \mathbb{R}$ lie to the left of $P$. By moving all points of $P$ one unit to the right ($\epsilon = 1$), the hyperplane depth of $q$ relative to $P$ increases by a factor of $n$, regardless of $k$. Thus HDD is not $k$-stable.
\end{proof}

\begin{theorem} \label{transformation}
The HDD function $D_P(q)$ relative to the set ${P}$ in $\mathbb{R}^d$ is not equivariant under affine transformations.
\end{theorem}
\begin{proof}
We will prove this theorem using a counter-example. Consider the set of points $P=\{p_0(0,0), p_1(4,0), p_2(2,1)\}$. Using Theorem~\ref{medianOnIntersections} we can show that the median is on point $p_2(2,1)$. Now consider the set $P'=\{p'_0(0,0), p'_1(4,0), p'_2(2,5)\}$ that is $P$ under the non-uniform affine transformation matrix $
\begin{bmatrix}
1 & 0 \\
0 & 5
\end{bmatrix}
$.
Using theorem \ref{medianOnIntersections} and \ref{simmetry} It can be shown that the median is on the line $p'_0p'_1$ now. This means that the HDD median is not equivariant under affine transformation. 
\end{proof}

As we now show, HDD is equivariant under similarity transformations, including translation, rotation, reflection and uniform scaling, since these preserve the shape of $P$.

\begin{theorem} \label{similarity}
The HDD function $D_P(q)$ relative to the set ${P}$ in $\mathbb{R}^d$ is equivariant under the similarity transformations.
\end{theorem}

\begin{proof}
For any rotation, reflection, or translation transformation $f$, the distance from the query point $q$ to any hyperplane $h_i$ remains unchanged. That is, for any point $q$ and any hyperplane $h_i$, $\dist(q,h_i) = \dist(f(q),f(h_i))$.

For any uniform scaling transformation $f$ with a scaling factor of $k$, distances between each pair of points will be multiplied by $k$ after the transformation. Therefore it is easy to show that, for any query point $q$, the HDD will be multiplied by $k$ after uniform transformation. Therefore, the median is equivariant under the uniform scaling transformation.
\end{proof}

\section{Algorithms} \label{Algorithms}
In this section, we provide three algorithms: a) to compute HDD depth queries in $O(d\log n)$ time after $O(n^{2d^2+2d})$ preprocessing time,  b) to find an HDD median point in $O(dn^{d^2}\log n)$ time, and c)  to find an approximate HDD median. Let $P$ be a set of $n$ points in $\mathbb{R}^d$, and let $H_P$ be the set of $\binom{n}{d}$ hyperplanes determined by $d$ point in $P$.

\subsection{HDD Query Algorithm}
\label{sec:HDDAlg}
The hyperplane distance depth of a query point $q$ relative to $P$ can be computed by directly evaluating Equation~\eqref{eq1} in $O\left(\binom{n}{d}\right) = O(n^d)$ time. We will present an algorithm that can calculate HDD in logarithmic time after preprocessing. 
First, to measure the HDD of $q$, we need to store some coefficients belonging to each polytope formed by hyperplanes in $H_P$.

Consider Equation~\eqref{eq2}. Let $S_P$ be the set of all minimal polytopes determined by the arrangement of
hyperplanes in $H_P$. For a query point $q_{k}$ in a polytope $s_k \in S_P$, the coefficients $g_{i,q}$ for $h_i \in H_P$ are the same. Therefore, for any points $q_{k}$ in $s_k$, we can simplify the summation in \eqref{eq2} in $O(n^d)$ time and find the $2$ coefficients $a_k$ and $b_k$ such that
\begin{equation}
\label{eqn:summation}
D_P(q_k) =\sum_{h_i \in H_P} g_{i,q_k} \frac{{w_i}.q_k+{b_i}}{\|{w_i}\|}= a_kq_k+b_k
\end{equation}
Using Euler's characteristic theorem we know that there are $O(n^{d^2})$ polytopes formed by the hyperplanes in $H$ e.g. in Figure~\ref{fig:4PointsExample} there are $18$ polytopes (faces) formed by the $6$ hyperplanes (lines). Therefore we will need $O(2n^{d^2}) \in O(n^{d^2})$ space and $O(n^{d^2}n^d) \in O(n^{d^2+d})$ time to preprocess.

Using the mentioned data structure we can calculate the HDD measure in $O(1)$ time if we know to which polytope the query point belongs.

Given $n$ hyperplanes in $d$-dimensional space and a query point $q$, it takes $O(\log n)$ time to find the $q$ location with a data structure of size $O(n^d)$ and a preprocessing time of $O(n^{2d+2})$\cite{CHAZELLE199453}. In our problem, there are $\binom{n}{d} \in O(n^d)$ hyperplanes. Therefore, with a preprocessing time of ${\binom{n}{d}}^{2d+2} \in O(n^{2d^2+2d})$ and a space of $O(n^{d^2})$, we can find the location of $q$ in $O\left(\log \binom{n}{d}\right) \in O(d\log n)$ time.

Now after finding the $q$'s location in $O(d\log n)$, we can calculate the HDD measure $D_P(q_k)$ in $O(1)$ using Equation~\eqref{eqn:summation}.

Therefore, after $O(n^{d^2+d}+n^{2d^2+2d}) \subseteq O(n^{2d^2+2d})$ preprocessing time using $O(n^{d^2})$ space, we can find the HDD of an arbitrary query point in $O(d\log n)$ time.
This proves the following theorem.

\begin{theorem}\label{thm:HDDAlg}
    We can preprocess any given set $P$ of $n$ points in $\mathbb{R}^d$ in $O(n^{2d^2+2d})$ time, such that given any point $q \in \mathbb{R}^d$, we can compute $D_P(q)$ in $O(d\log n)$ time.
\end{theorem}

\subsection{Finding a HDD Median}
\label{sec:algorithms.medianExact}
By Theorem~\ref{medianOnIntersections}, 
a straightforward algorithm for finding an HDD median of $P$ is to check all points of intersection between $d$ hyperplanes in $H_P$ using an exhaustive search. There are $\binom{n}{d}$ hyperplanes in $H_P$ and therefore $\binom{\binom{n}{d}}{d} \in O(n^{d^2})$ intersection points between hyperplanes in $H_P$. Since it takes $O(n^d)$ to compute the Equation~\eqref{eq1} directly, a HDD median of $P$ can be found in $O(n^{d^2+d})$ time by this brute-force algorithm.

Next, we will introduce an algorithm that finds the HHD median in $O(dn^{d^2} \log n)$ time. 
When $d=2$, this second algorithm runs in
$O(n^4\log n)$ time, compared to $O(n^6)$ time for the brute-force algorithm. 
First, we will show that we can find the point with the smallest HDD on a line in $O(d n^d \log n)$ time. Consider the intersection of $d-1$ hyperplanes in $H_P$ that determine a line $\ell$. Since every hyperplane in $H_P$ has exactly one point of intersection with $\ell$, $H_P \cap \ell$ is a set of $O(n^d)$ points of intersection. By Theorem~\ref{convexity}, we can conclude that the hyperplane depth of points on $\ell$ is a convex function. Since $H_P \cap \ell$ is discrete, using binary search and calculating HDD in $O(n^d)$ time using Equation~\eqref{eq1}, we can find the intersection point with the minimum HDD in $O(n^d \log (n^d)) = O(d n^d \log n)$ time.
 
 We can use the algorithm above to find the minimum point for each intersection line among hyperplanes in $H_P$ to find an HDD median. Since each $d-1$ hyperplanes in $H_P$ form a line, there are $\binom{\binom{n}{d}}{d-1} \in O(n^{d^2-d})$ lines and thus we can find the median in $O(d n^{d^2} \log n)$ time.
This proves the following theorem.

\begin{theorem}\label{thm:medianExact}
    Given a set $P$ of $n$ points in $\mathbb{R}^d$, we can find an HDD median of $P$ in $O(dn^{d^2}\log n)$ time.
\end{theorem}

\subsection{Finding an Approximate HDD Median in \boldmath$\mathbb{R}^2$}
\label{sec:algorithms.medianApproximate2D}
In this section, we will present an approximation algorithm to find an HDD median of $P$ with an error of $\frac{a\sqrt{2}}{2^{\frac{m}{2}+1}}$ in $O(mn^2\log n)$ time, for any fixed $m \in \mathbb{Z}^+$, where $a$ is the diameter of $P$. 

\begin{theorem}\label{thm:medianApproximate2D}
    Given a set $P$ of $n$ points in $\mathbb{R}^2$, in $O(mn^2\log n)$ time we can find a point $x'$ in $\mathbb{R}^2$ such that $\dist(x',x) \leq \frac{a\sqrt{2}}{2^{\frac{m}{2}+1}}$, for any fixed $m \in \mathbb{Z}^+$, where $x$ denotes an HDD median of $P$ and $a = \max_{p,q\in P}\dist(p,q)$. 
\end{theorem}

\begin{proof}
Let $l_a$ be an arbitrary line among the lines in $H_P$ (see Figure~\ref{fig:eliminate2d}). There are $n \choose 2$ lines in $H_P$ and, consequently, $O(n^2)$ points of intersection between $l_a$ and lines in $H_P$. Using an analogous argument as in the proof of Theorem~\ref{medianOnIntersections}, the point with minimum HDD on $l_a$ lies at an intersection of $l_a$ and a line in $H_P$. Therefore, using the same algorithm described in Section~\ref{sec:algorithms.medianExact}, we can find the point $i_{\min}$ on $l_a$ with minimum HDD in $O(n^2\log n)$ time; let $h_{\min}$ denote the line in $H_P$ such that $i_{\min} = h_{\min} \cap l_a$. Next, we find the closest points of intersection in $H_P$ to $i_{\min}$ on the line $h_{\min}$ in each direction, say $I_u$ and $I_d$. We compute the HDD for all the three points $i_{\min}$, $I_d$, and $I_u$. Since $i_{\min}$ has minimum HDD on the line $l_a$, if $D_P(i_{\min}) < \min\{D_P(I_u), D_P(I_d) \}$, then $i_{\min}$ is the HDD median (by Theorem~\ref{convexity}). By Theorem~\ref{convexity} again, $D_P(i_{\min}) < D_P(I_u)$ or $D_P(i_{\min}) < D_P(I_d)$. Furthermore, $D_P(i_{\min}) > D_P(I_u)$ or $D_P(i_{\min}) > D_P(I_d)$. Without loss of generality, suppose $D_P(I_d)<D_P(i_{\min})<D_P(I_u)$. We claim that all points in the half-plane bounded by $h_{\min}$ that contains $I_u$ have HDD that exceeds $D_P(i_{\min})$; we prove this by contradiction. Suppose there exists a point $A$ in this half-plane such that $D_P(A)<D_P(i_{\min})$. Let $B$ be the intersection point of the line $l_a$ and the line segment $\overline{AI_d}$. Since $i_{\min}$ has minimum HDD on the line $l_a$, therefore, $D_P(i_{\min})<D_P(B)$. Furthermore, $D_P(A)<D_P(B)$. On the other hand, we assumed $D_P(I_d)<D_P(i_{\min})<D_P(I_u)$ and we know $D_P(i_{\min})<D_P(B)$. Consequently, $D_P(I_d)<D_P(B)$. Combining the two resulting inequalities above, we have $D_P(A)<D_P(B)$ and $D_P(I_d)<D_P(B)$, which is impossible since the three points are on the same line and the HDD function is convex. Therefore, no such point $A$ can exist. 

\begin{figure}[htp]
    \centering
    \includegraphics[width=8cm]{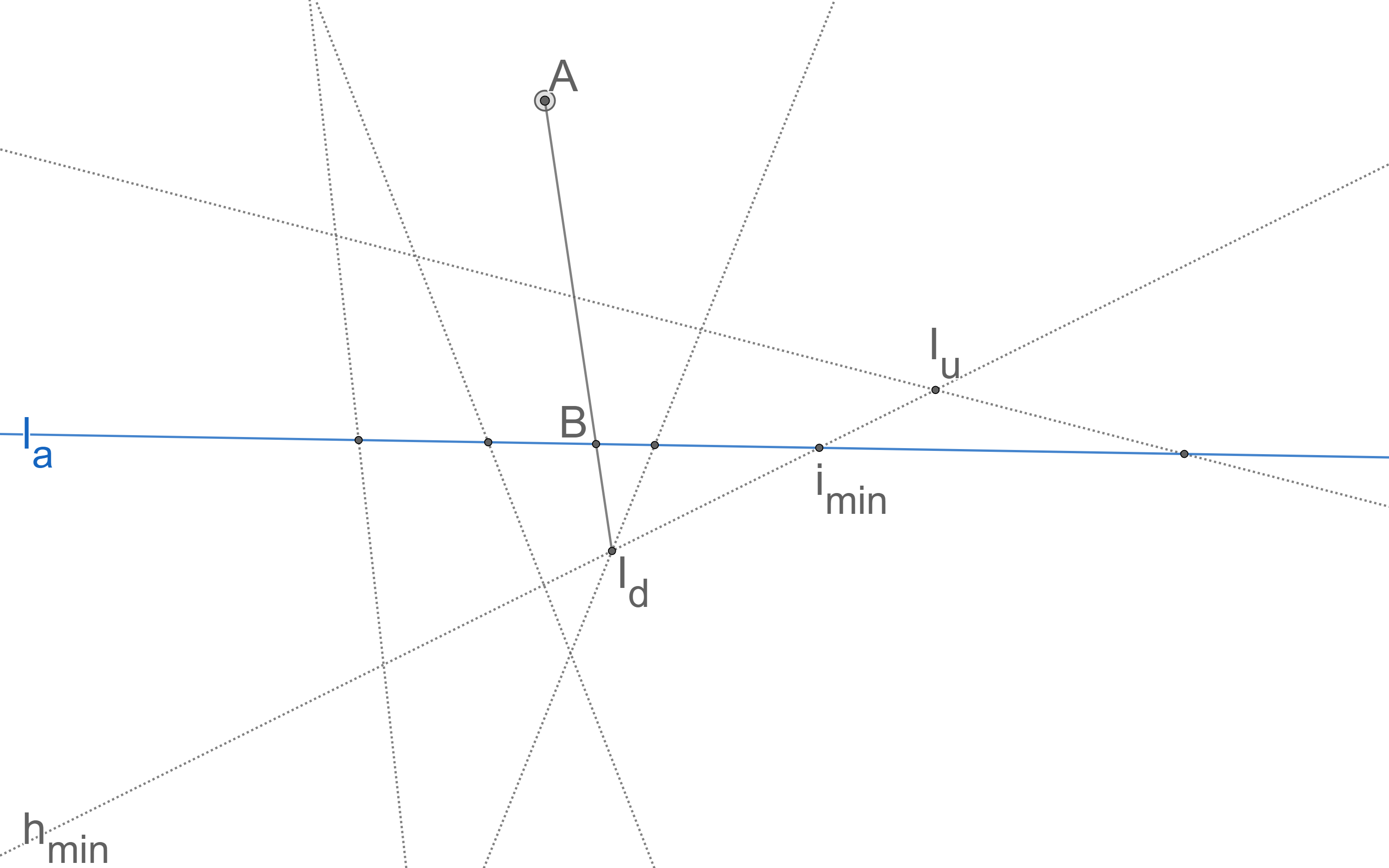}
    \caption{An algorithm to eliminate the points belonging to a half-space. The blue line $l_a$ is an arbitrary line dividing the space into 2 halves. The dotted gray lines are the lines in $H_P$.}
    \label{fig:eliminate2d}
\end{figure}
Therefore, no point of intersection in $H_P$ in this half-plane can be an HDD median of $P$; in $O(n^2\log n + 3n^2) \subseteq O(n^2\log n)$ time we can remove these points from consideration in our search for a median. 

Now we will use this property to approximate the median point. Firstly, we will find the diameter $a$ of the input points in $O(n)$ time and consider an $a \times a$ square that contains $P$ (see Figure~\ref{fig:square}). By Theorem~\ref{MedianInCH}, we know that the median lies inside this square. At each step, we draw the two lines $ON$ and $OM$ that partition the square into four similar smaller squares, each with dimensions $\frac{a}{2} \times \frac{a}{2}$, and we apply the above algorithm to eliminate two half-planes in $O(n^d \log n)$ time. After $m$ steps we have a square of dimensions $\frac{a}{2^m} \times \frac{a}{2^m}$ and we return its center as an approximation of the HDD median. Since the HDD median is a point inside this square, the error is at most $\frac{a\sqrt{2}}{2^{m+1}}$. The total time complexity of the algorithm is $O(mn^2\log n)$.
\end{proof}
\begin{figure}[htp]
    \centering
    \includegraphics[width=8cm]{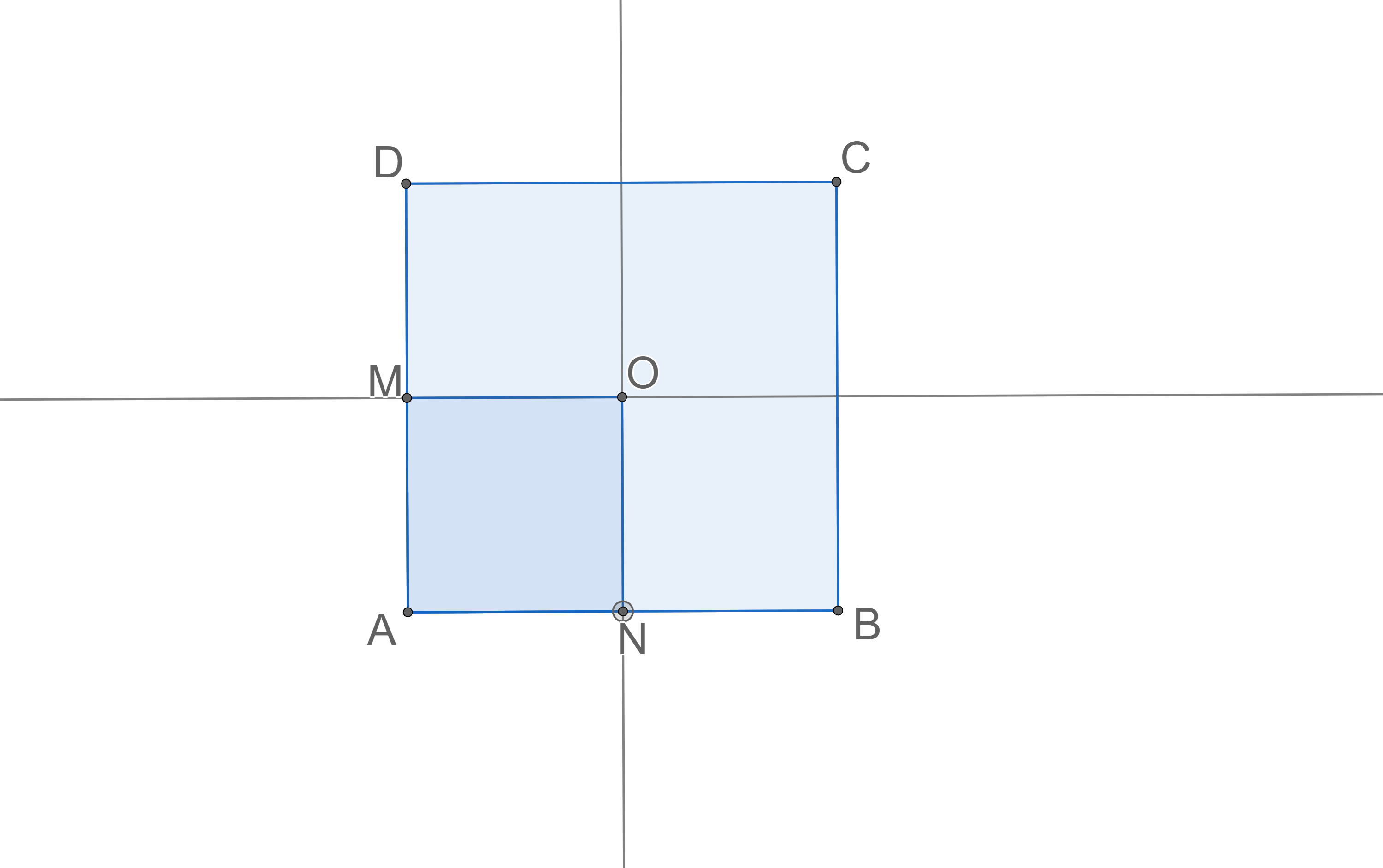}
    \caption{Illustration in support of Theorem~\ref{thm:medianApproximate2D}}
    \label{fig:square}
\end{figure}

This strategy can be generalized to higher dimensions by finding the minimum HDD on an arbitrary hyperplane $h_a$ (analogous to the line $l_a$ in the proof of Theorem~\ref{thm:medianApproximate2D}) to eliminate a half-space, but the time complexity of finding the minimum HDD point on $h_a$ is high.

\section{Discussion and Possible Directions for Future Research}
\label{sec:conclusion}
Our algorithm for computing HDD queries presented in Section~\ref{sec:HDDAlg} requires $O(n^{d^2})$ space and $O(n^{2d^2+2d})$ preprocessing time. One natural possible direction for future research is to identify algorithms with improved preprocessing time or space.

Our algorithm for finding an HDD median presented in Section~\ref{sec:algorithms.medianExact} requires $O(d n^{d^2} \log n)$ time. In addition to seeking to identify lower bounds on the worst-case running time required to find an HDD median, we could attempt to reduce the running time using techniques such as gradient descent or linear programming.

Our analysis of our algorithm for finding an approximate HDD median presented in Section~\ref{sec:algorithms.medianApproximate2D} does not capitalize on the fact that the number of candidate points decreases on each step; we charge $O(n^2\log n)$ time per step. If it could be shown that a constant fraction of the remaining points are eliminated on each step, then the bound on the algorithm's time complexity would be significantly improved. 

Finally, we could consider alternative definitions for depth using similar notions to those in Definition~\ref{hddDefinition}. E.g., one can define a ``line distance depth'' that evaluates the distances to all possible lines passing through each pair of points in the set of input points. This definition coincides with Definition~\ref{hddDefinition} when $d \leq 2$, but differs in higher dimensions, for $d \geq 3$.


\small
\bibliographystyle{abbrv}

\bibliography{refs}

\begin{thebibliography}{10}

\bibitem{bajaj1988}
C.~Bajaj.
\newblock The algebraic degree of geometric optimization problems.
\newblock {\em Discrete and Computational Geometry}, 3:177--191, 1988.

\bibitem{barnett1976ordering}
V.~Barnett.
\newblock The ordering of multivariate data.
\newblock {\em Journal of the Royal Statistical Society: Series A (General)}, 139(3):318--344, 1976.

\bibitem{CHAZELLE199453}
B.~Chazelle and J.~Friedman.
\newblock Point location among hyperplanes and unidirectional ray-shooting.
\newblock {\em Computational Geometry}, 4(2):53--62, 1994.

\bibitem{durier1985geometrical}
R.~Durier and C.~Michelot.
\newblock Geometrical properties of the {Fermat-Weber} problem.
\newblock {\em European Journal of Operational Research}, 20(3):332--343, 1985.

\bibitem{durocher2009}
S.~Durocher and D.~Kirkpatrick.
\newblock The projection median of a set of points.
\newblock {\em Computational Geometry: Theory and Applications}, 42(5):364--375, 2009.

\bibitem{durocher2017}
S.~Durocher, A.~Leblanc, and M.~Skala.
\newblock The projection median as a weighted average.
\newblock {\em Journal of Computational Geometry}, 8(1):78--104, 2017.

\bibitem{farahani2010multiple}
R.~Z. Farahani, M.~SteadieSeifi, and N.~Asgari.
\newblock Multiple criteria facility location problems: A survey.
\newblock {\em Applied mathematical modelling}, 34(7):1689--1709, 2010.

\bibitem{liu1990notion}
R.~Y. Liu.
\newblock On a notion of data depth based on random simplices.
\newblock {\em The Annals of Statistics}, pages 405--414, 1990.

\bibitem{mahalanobis2018generalized}
P.~C. Mahalanobis.
\newblock On the generalized distance in statistics.
\newblock {\em Sankhy{\=a}: The Indian Journal of Statistics, Series A (2008-)}, 80:S1--S7, 2018.

\bibitem{murray1998cluster}
A.~T. Murray and V.~Estivill-Castro.
\newblock Cluster discovery techniques for exploratory spatial data analysis.
\newblock {\em International journal of geographical information science}, 12(5):431--443, 1998.

\bibitem{oja1983descriptive}
H.~Oja.
\newblock Descriptive statistics for multivariate distributions.
\newblock {\em Statistics \& Probability Letters}, 1(6):327--332, 1983.

\bibitem{tukey1975}
J.~Tukey.
\newblock Mathematics and the picturing of data.
\newblock In {\em Proc. Int. Cong. Math.}, pages 523--531, 1957.

\bibitem{vardi2000multivariate}
Y.~Vardi and C.-H. Zhang.
\newblock The multivariate $l_1$-median and associated data depth.
\newblock {\em Proceedings of the National Academy of Sciences}, 97(4):1423--1426, 2000.

\bibitem{zuo2000general}
Y.~Zuo and R.~Serfling.
\newblock General notions of statistical depth function.
\newblock {\em Annals of statistics}, 28(2):461--482, 2000.

\end{thebibliography}


\end{document}